\documentclass{acmart}
 \pdfoutput=1
 \usepackage[position=top]{caption}
\makeatletter
 \newif\if@restonecol

\usepackage{amsmath}
\makeatletter
 \newif\if@restonecol

\usepackage{graphicx}
\usepackage{epstopdf}
\usepackage{subfig}

\usepackage[linesnumbered,ruled,vlined]{algorithm2e}
\usepackage{algpseudocode}

\acmConference[]{}

\newtheorem{definition}{Definition} 
\newtheorem{lemma}{Lemma}
\newtheorem{example}{Example}

\newcommand{\tabincell}[2]{\begin{tabular}{@{}#1@{}}#2\end{tabular}}

\begin{document}

\title{BigCarl: Mining frequent subnets from a single large Petri net}

\author{Ruqian Lu}
\authornote{Corresponding author}
\email{rqlu@math.ac.cn}

\affiliation{%
  \institution{Academy of Mathematics and Systems Science Key Lab of MADIS Chinese Academy Science,}
    \institution{Key Laborartory of Intelligent Information Processing of Chinese Academy Sciences}
  \city{Beijing}
  \country{China}
  \postcode{100190}
}

\author{Shuhan Zhang}
\affiliation{%
\institution{Key Laborartory of Intelligent Information Processing of Chinese Academy Sciences,}
\institution{University of Chinese Academy Sciences}
  \city{Beijing}
  \country{China}}
\email{zhangshuhan@ict.ac.cn}


\begin{abstract}
While there have been lots of work studying frequent subgraph mining, very few publications have discussed frequent subnet mining from more complicated data structures such as Petri nets. This paper studies frequent subnets mining from a single large Petri net. We follow the idea of transforming a Petri net in net graph form and to mine frequent sub-net graphs to avoid high complexity. Technically, we take a minimal traversal approach to produce a canonical label of the big net graph. We adapted the maximal independent embedding set approach to the net graph representation and proposed an incremental pattern growth (independent embedding set reduction) way for discovering frequent sub-net graphs from the single large net graph, which are finally transformed back to frequent subnets. Extensive performance studies made on a single large Petri net, which contains 10 K events, 40 K conditions and 30 K arcs, showed that our approach is correct and the complexity is reasonable.
\end{abstract}




\maketitle

\section{Introduction}

Frequent subgraph mining has been a hot research topic since the end of the last century \cite{jiang_survey_2013}. However, before the discovery of the pattern growth approach, in particular the proposal of gSpan algorithm \cite{xifeng_yan_gspan:_2002}, frequent subgraph mining was a very hard and tough task \cite{cook_mining_2007}. This task became even more difficult when scientists tried to mine frequent subgraphs from a single large graph \cite{kuramochi_finding_2005}. It invoked many research works \cite{kuramochi_finding_2005,fiedler_support_2007,washio_what_2008,elseidy_grami:_2014,meng_counting_2020}. The main difficulty is how to count the number of embedding copies of a frequent subgraph pattern in case that they are overlapped. To solve this problem, many techniques have been proposed. However, most of them have high computational complexity. For example, MIS \cite{kuramochi_finding_2005} and HO  \cite{fiedler_support_2007} are all NP-Hard. MNI \cite{elseidy_grami:_2014} is polynomial, however, it is heuristic.

Besides that, another phenomenon has raised our attention. While there have been lots of work studying frequent subgraph mining, very few publications have discussed frequent subnet mining from more complicated data structures such as Petri nets \cite{reisig_lectures_1998,murata_petri_1989}. It was not until recently that a preprint has been published \cite{lu_pspan_nodate}, where frequent subnet mining from a set of Petri nets is studied.

The theory of Petri nets, or net theory for short, is a powerful theory for concurrent distributed systems \cite{yuan_application_2013}. It was founded by C.A. Petri in his seminal thesis 'Kommunikation mit Automaten' in 1962 \cite{petri_kommunikation_1962}. Since then the net theory has experienced a rash development and found a wide scope of application, including machinery, chemistry, biology, linguistics, workflow analysis, and many others \cite{hu_petri_2015,liu_robustness_2013,khan_survey_2018,silva_half_2013,bernardo_introduction_2007,peterson_petri_1981,tax_localprocessmodeldiscovery:_2018}. Researches done in the past nearly 60 years have accumulated a huge repository of Petri net pieces. A deep analysis of this repository is needed. However, as we said above, we have found a serious lack of subnet mining study results. This situation has motivated us to do the research presented in the present paper, which studies frequent subnet mining from a single large Petri net.

Technically, we adopt the strategy of transforming a Petri net into net graph representation introduced in \cite{lu_pspan_nodate} and doing frequent sub-net graph mining based on it, but with some modifications and extensions, which will be explained in section \ref{sec:net}. To start sub-net graph mining, we first take a minimal traversal approach to produce a minimal DFS code of the net graph, which represents a canonical labeling of the net graph and is slightly different from \cite{lu_pspan_nodate}. The frequent sub-net graph mining is based on the maximal independent set support principle. The discovery procedure goes ahead in an incremental way. Finally, the discovered frequent sub-net graphs and their distinct embedding copies will be transformed back to Petri net form. At the end, extensive performance studies made on a single large Petri net, which contains 10 K events, 40 K conditions and 30 K arcs, showed that our approach is correct and the complexity is reasonable.

The remainder of the paper is arranged as follows: Section \ref{sec:preli} presents the preliminaries of complete subnet and the net graph representation. Section \ref{sec:minimal} introduces the minimal traversal of a net graph. Section \ref{sec:max}  explains the main idea of BigCarl algorithm under the maximal independent set support framework. Section ~\ref{sec:algorithm} presents in detail the BigCarl algorithm Section \ref{sec:evaluation} evaluates the net graph representation and BigCarl algorithm with a single large Petri net generated in a random way. Finally, section \ref{sec:conclude} provides some concluding remarks.

\section{ PRELIMINARY}\label{sec:preli}

\subsection{Complete subnets}

Petri nets differ from graphs not only in their topological structure (i.e., Petri nets have two types of nodes: the event nodes and the condition nodes) \cite{murata_petri_1989}. According to the point of view of C.A.Petri himself, a subnet of a Petri net should be semantically complete \cite{peterson_petri_1977}. This means if you cut off a subnet from a Petri net, then together with an event node, you should also cut off all its input and output condition nodes. This point of view plays an important role in defining frequent subnet mining. This paper extends the concept of complete subnets in \cite{lu_pspan_nodate} by differentiating two types of complete subnets, where the e-type complete nets coincide with the complete nets defined in  \cite{lu_pspan_nodate}.

\begin{definition}[e-type/c-type complete subnet]

1. Given a Petri net $N$. If $e(c)$ is an event (condition) node of $N$, then $e(c)$ together with all its input and output conditions (events) is called an \textbf{e-type (c-type) 1-complete subnet}. \\
2. An e-type (c-type) m-complete subnet is composed by identifying some of the condition nodes (event nodes) of the m 1-complete subnets with each other.
\end{definition}
$\hfill \square$

Our algorithm BigCarl mines only e-type frequent complete subnets. Since there are many different sorts of Petri nets, when we talk about Petri net in this paper, we always mean its basic sort: the pure C/E net, except if some special interpretation is given. Following is the definition.

\begin{example}
Figure ~\ref{a} borrows a Petri net from \cite{lu_pspan_nodate}, ~\ref{b} and ~\ref{c} show an e-type respectively c-type 1-complete subnet of ~\ref{a}. Figure ~\ref{d} and ~\ref{e} show two different connections of ~\ref{b} and ~\ref{c}, where ~\ref{d} is an e-type 2-complete subnet while ~\ref{e} is a c-type 2-complete subnet.
\end{example}

\begin{figure}
%

	\begin{minipage}[c]{\linewidth}
		\centering
		\subfloat[A pure C/E net ]{\includegraphics[width=0.6\textwidth]{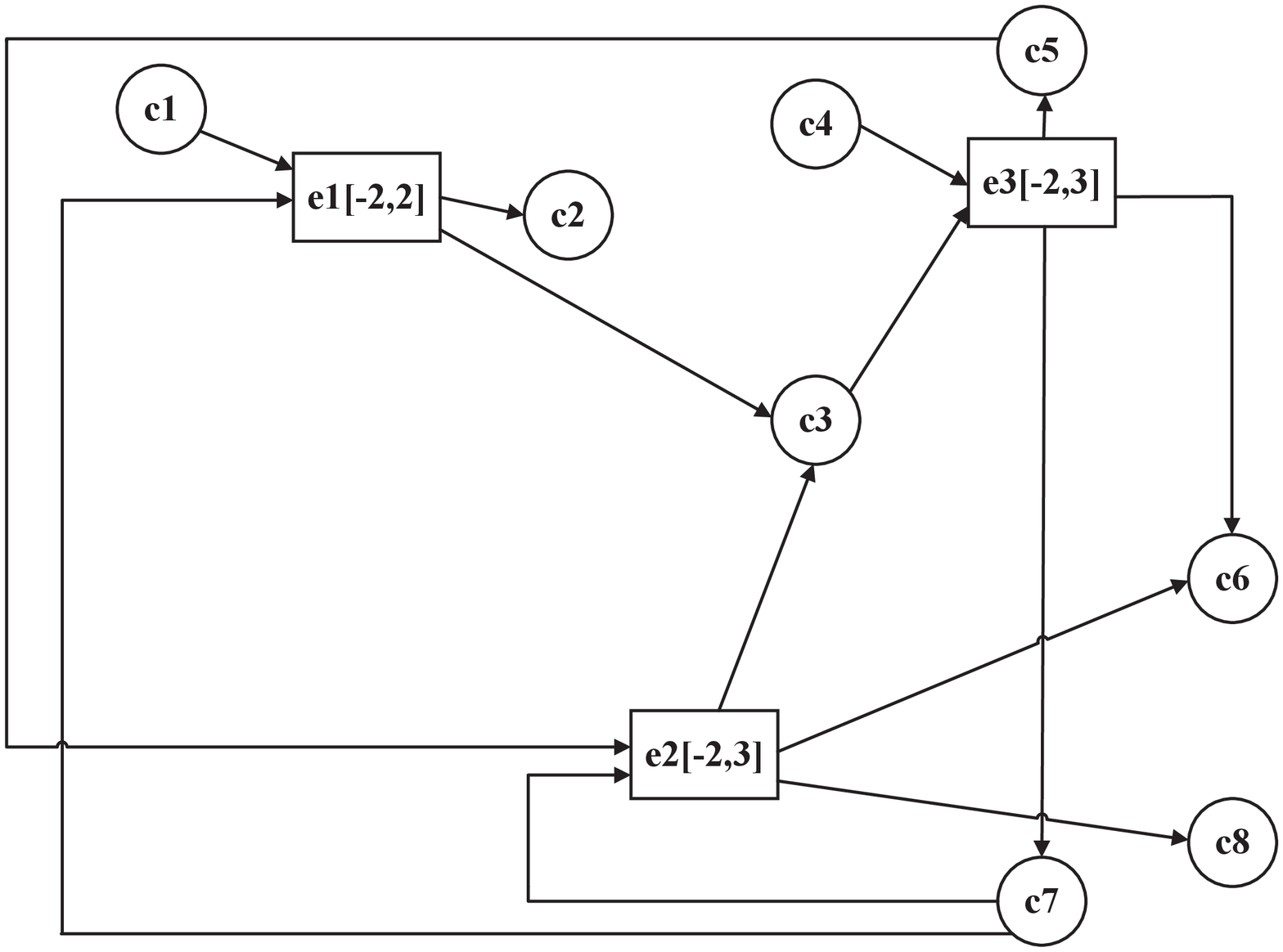}\label{a}}

	\end{minipage}
		\begin{minipage}[c]{0.4\linewidth}
			\centering
			\subfloat[An e-type 1-complete subnet ]{\includegraphics[width=0.8\textwidth]{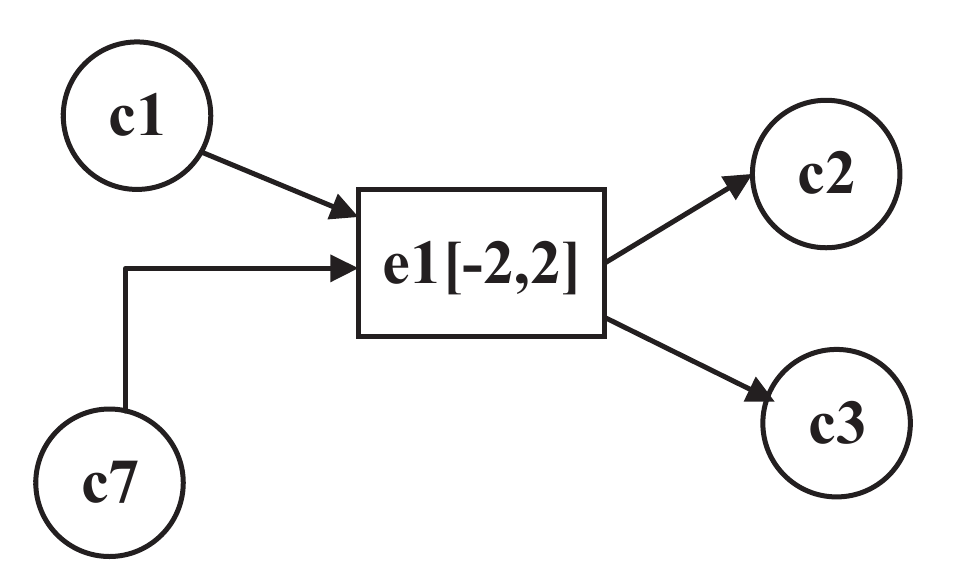}\label{b}}

		\end{minipage}
				\begin{minipage}[c]{0.4\linewidth}
					\centering
					\subfloat[A c-type 1-complete subnet ]{\includegraphics[width=0.8\textwidth]{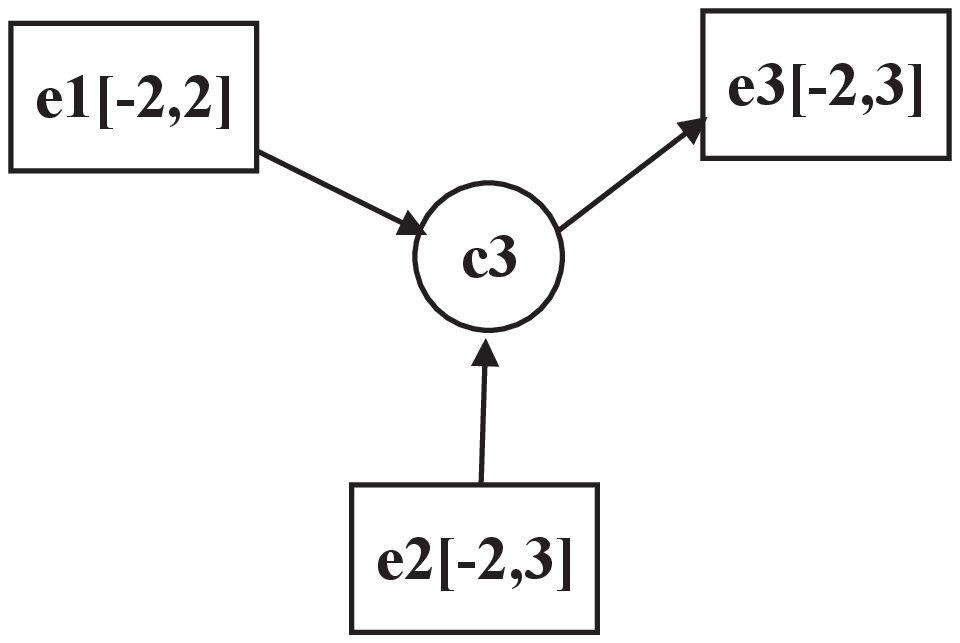}\label{c}}

				\end{minipage}
	\begin{minipage}[c]{.4\linewidth}
		\centering
		\subfloat[An e-type 2-complete subnet composed by $e_1$ and $e_3$ ]{\includegraphics[width=\textwidth]{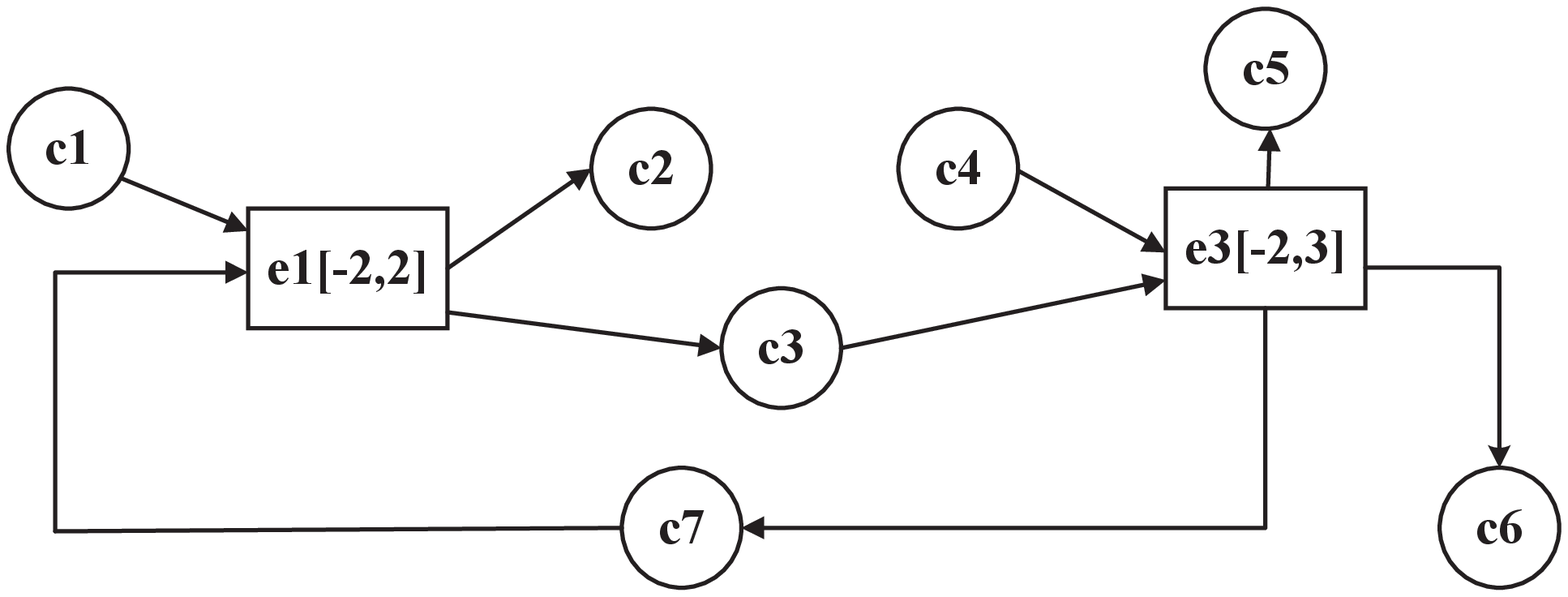}\label{d}
		}

	\end{minipage}
	\begin{minipage}[c]{.4\linewidth}
		\centering
		\subfloat[A c-type 2-complete subnet composed by $c_3$ and $c_5$ ]{\includegraphics[width=0.9\textwidth]{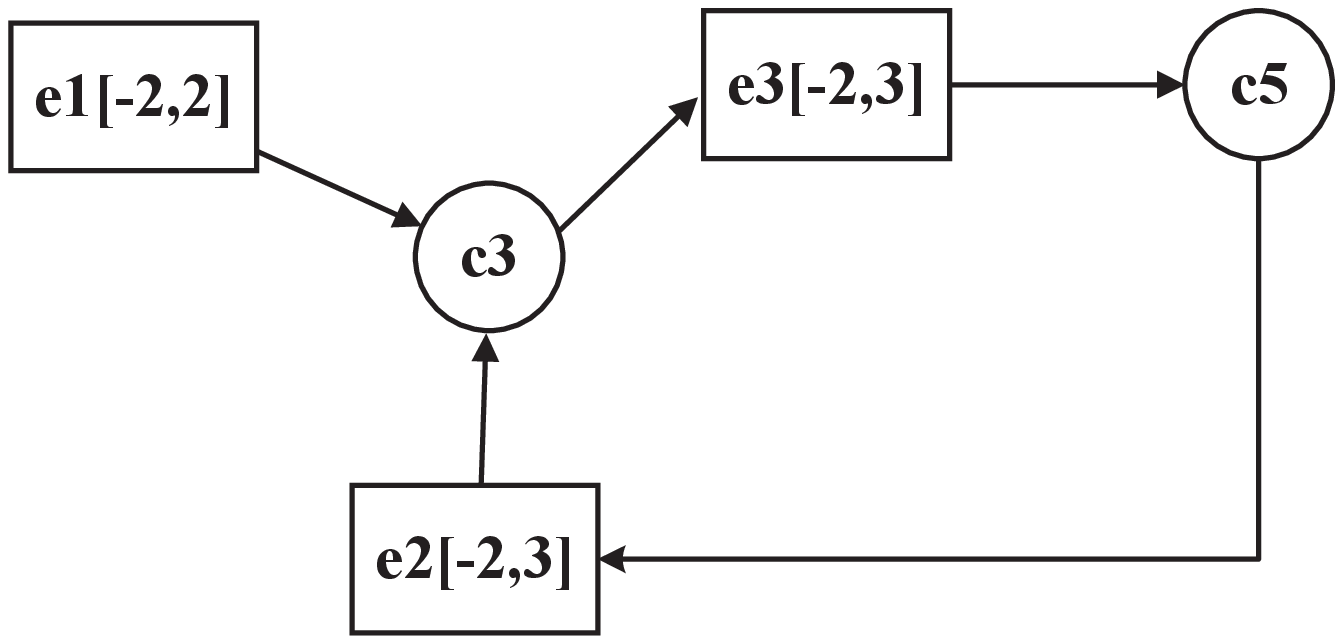}\label{e}
		}

	\end{minipage}

\caption{A Petri net and its complete subnets }
\label{1}
\end{figure}

\subsection{Net graph representation}\label{sec:net}

Net graph is a graph-like representation introduced in  \cite{lu_pspan_nodate} for reducing the complexity of Petri nets’ frequent subnet mining. According to the e-type complete subnet and c-type complete subnet introduced above, we differentiate in this paper between two types of net graphs: the e-type and c-type net graphs.

\begin{definition}[e-type net graph] \label{e-type n}
	A e-type net graph $NGE=(V_e,C_e;W_e)$ is a pseudo-graph transformed from a pure C/E net $N=(C,E;F)$, where\\
	1. $V_e$ is the set of NGE's nodes, $C_e=C$. $V_e=V+$tagging, where each $v\in V_e$ is represented as $e<$signed sequence of conditions connected to $e$ in $N$$>$ for some $e$ of $E$. Here $e$ plays the role of $v$'s node name. For any condition $c$ in this sequence, the sign is '-'('+') if $F$ contains an arc from $c(e)$ to $e(c)$. Note that in this sequence, the signed conditions are alphabetically ordered where symbol ‘-’ is always before symbol ‘+’;   \\
	2. $W_e$ is the set of  NGE's edges. Each edge $w$ connecting $v_1$ and $v_2$ has a tagging which is a lexicographically ordered sequence (symbol ‘-’ is always before symbol ‘+’) of triples where each triple has the form
 $\{h_1,c,h_2\}$, where $c$ is a condition of $C_e$ connecting the two nodes $e_1$ and $e_2$ of 
 $E$, which are the node names of $v_1$ and $v_2$. ${h_1}='-'('+')$ means $F$ contains an arc from $c(e_1)$ to $e_1(c)$. The same for ${h_2}$ and $e_2$. 
\end{definition}
$\hfill \square$
\begin{definition}[c-type net graph]
A c-type net graph $NGC=(C_c,E_c; W_c)$ is a pseudo-graph transformed from a pure C/E net $N=(C,E;F)$. All statements of definition ~\ref{e-type n} apply to $NGC$ if we exchange the roles of events and conditions. 
\end{definition}
$\hfill \square$

The following lemma shows that there is a dual relation between the two types of net graphs.

\begin{lemma}
Given a C/E net $N$. Transform it to its dual form $N'$ by renaming each event as a condition and each condition as an event while keeping the arcs unchanged. Construct $N$'s net graph $NG'$. Transform $NG'$ to its dual form $NG''$ by renaming each node as a condition and each condition in the node's tagging and edge's tagging as an event. At the same time reverse the signs in $G'$'s tagging. Then $NG''$ is equal to  $N$'s net graph $NG$. 
\end{lemma}
\begin{proof}
Omitted.
\end{proof}

\begin{example}
Figure ~\ref{2} shows the e-type net graph and c-type net graph transformed from the pure C/E net in figure ~\ref{a}.
\end{example}
\begin{figure}
%


		\begin{minipage}[c]{0.8\linewidth}
			\centering
			\subfloat[An e-type net graph transformed from figure ~\ref{a} ]{\includegraphics[width=0.5\textwidth]{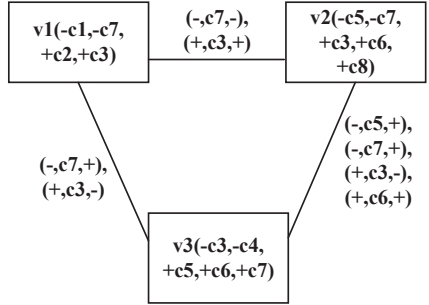}\label{enet}}

		\end{minipage}
				\begin{minipage}[c]{0.8\linewidth}
					\centering
					\subfloat[A c-type net graph transformed from figure ~\ref{a}  ]{\includegraphics[width=\textwidth]{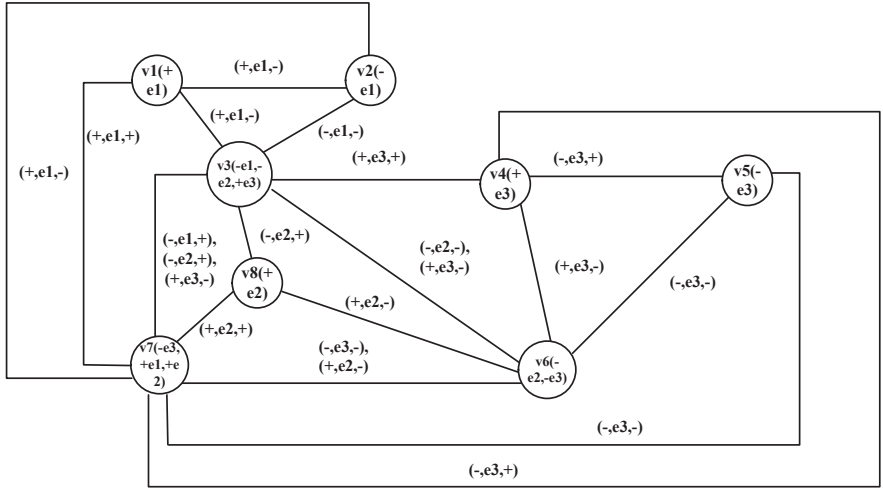}\label{cnet}}

				\end{minipage}

\caption{An e-type net graph and a c-type net graph }
\label{2}
\end{figure}

\begin{lemma}
Given any pure C/E net $N$, there is a one-one transformation between $N$ and its e-type net graph representation and its c-type graph representation. 
\end{lemma}
\begin{proof}
Omitted.
\end{proof}
\section{MINIMAL TRAVERSAL of a NET GRAPH}\label{sec:minimal}

To start with frequent sub-net graph mining, similar to the idea of \cite{xifeng_yan_gspan:_2002}, we first have to define a depth first traversal of the whole net graph which produces a canonical labeling of the net graph.. For that purpose, we need a lexicographic order of DFS codes. Following the idea of \cite{lu_pspan_nodate}, we do not sort the DFS codes in the lexicographic order after the algorithm has produced all DFS codes, but produce a minimal DFS code during a specific way of net graph traversal, the minimal traversal. However, the strategy taken in this paper is slightly different from that used in \cite{xifeng_yan_gspan:_2002}. In order to assure the uniqueness of minimal traversal, we introduce a natural rule in the following for C/E net construction. 
\begin{definition}[clear C/E net]
A C/E net is called clear, if for any event node $e$,\\
1. No two of $e$'s input condition nodes have the same name, and\\
2. No two of $e$'s output condition nodes have the same name.

\end{definition}
$\hfill \square$

This concept is semantically sound since it is meaningless to have repeated same condition nodes in an event node’s preset or postset. It can be generalized to other kinds of Petri nets. But for the moment we are only interested in C/E nets. We assume all C/E nets discussed in the remaining part of this paper are clear nets.

The following algorithm is a combination of gSpan’s idea about depth first travel and its idea about minimal traversal and our clear net definition above.

%
%
%
%
%
%

\begin{algorithm}[H]
	\caption{BigCarl-Minimal-DFS-Traversal ($NGE$: e-type net graph)}
	\label{algo:mini}

	\raggedright
	\KwIn{$NGE$: e-type net graph}
		\LinesNumbered
	$i=0$;\\
	Choose the minimal node $u$ (node with minimal tagging) as the starting of traversal;\\
	Mark $u$ as current node and already visited, put $u$ in $MinE[0]$; \tcp{$MinE[0]$ is the set of all nodes of net graph with their taggings}
	\While{exist edges untraversed}{
	\If{backward edge $w_1$’s another node $v_1$ is minimal}{put $w_1$ in $MinE[1]$, put $v_1$ in $MinE[0]$, make $v_1$ as current node; \tcp{$MinE[1]$ is the set of all 1-edges of net graph}
	
	}
	\Else{put minimal backward edge $w_2$(edge with minimal tagging) in $MinE[1]$;\\
	put $w_2$’s another node $v_2$ in $MinE[0]$, make $v_2$ as current node;
	}
	$i=i+1$;\\
	\If{there are untraversed forward edges from current node}{put minimal forward edge $w_3$ in $MinE[1]$, put $v_3$ in $MinE[0]$;\\
	make $w_3$'s another node $v_3$ as current node, $i= i +1$;
	}
	} 
$NoE=i;$\tcp{$NoE$ is the id of edge when traversing}
\end{algorithm}
\begin{lemma}
The decision of minima travel path in algorithm 1 is uniquely determined.
\end{lemma}
\begin{proof}
For backwards traversal, the minimal node on the next ends of the backward edges is uniquely determined since all these end nodes have been traversed at least once and the least traverse number is uniquely determined. For forwards traversal, the forward edges must have different tagging since the conditions contained in the tagging of different forward edges must have different names according to the clear net principle. 
\end{proof}
\section{MAXIMUM INDEPENDENT SET SUPPORT}\label{sec:max}

\subsection{Subnet embedding overlap}
For developing the BigCarl algorithm we first introduce two definitions about overlap.

\begin{definition}
1. Two Petri nets have an \textbf{e-type overlap} if both sides have at least one event node in common or a \textbf{c-type overlap} if both sides have at least one condition node in common.\\
2. With the overlap concept introduced above, we can redefine m-complete subnets as follows: the connection of m e-type 1-complete subnets with help of c-type overlap is called an e-type m-complete subnet, while the connection of m c-type 1-complete subnets with help of e-type overlap is called a c-type m-complete subnet. The connect action itself is called e-connect respectively c-connect. 

\end{definition}
$\hfill \square$

\begin{lemma}
Assume a pure C/E net $N$ is transformed to its e-type (c-type) net graph $NG$, where $N$’s e-type (c-type) complete nets $N_1'$ and $N_2'$ are transformed to two sub-net graphs $NG_1'$ and $NG_2'$. Then $N_1'$ and $N_2'$ do not have event (condition) node overlap if and only if $NG_1'$ and $NG_2'$ do not have node overlap.
\end{lemma}
\begin{proof}
Straightforward. 
\end{proof}
Our BigCarl algorithm adopts the MIS (Maximum Independent Set support) approach with a modification. While a standard MIS algorithm searches subgraph embedding without edge overlap, BigCarl searches subnet embedding without overlap of some node type. Due to semantic reason explained above, BigCarl either only mines e-type m-complete subnets, or only mines c-type m-complete subnets for any m. In the former case, c-type overlap is admitted as junction of several e-type 1-complete nets, while in the latter case, e-type overlap is admitted as junction of several c-type 1-complete nets. It can be proved that without such junctions frequent subnet mining would be incomplete.
\subsection{Bulding a Maximum Independent Set}

\begin{definition}[level-k frequent sub-net graph]
Given a frequency threshold $MinSup>0$ and a net graph $NG$. $Ng'$ is a connected sub-net graph. We call $Ng'$ a frequent sub-net graph if $Ng'$ has $k$ edges and if there are no less than $MinSup$ embedding copies in $NG$ which have no nodes in common.
\end{definition}
$\hfill \square$

The discovery of frequent sub-net graphs is made in an incremental way. A level-k sub-net graph is one which contains k edges. Different from usual subgraph mining, a frequent level-0 sub-net graph is meaningful since it represents a 1-complete subnet (subnet with one event node and all condition nodes connected to it). Our algorithm BigCarl finds all level-0 frequent sub-net graphs with their maximal independent embedding set. BigCarl then tries to increase the frequent sub-net graphs level by level. In this process it keeps a maximal independent set of embedding for each candidate frequent sub-net graph. This set reduces gradually when the above pattern grows until the largest possible $k$ is reached.

We have also extended our approach to solve frequent subnet mining under c-type overlap and mixed e-c-type overlap. Note that although we just take the C/E net as an example model to study the problem of frequent subnet mining from a single large net, we have proved that this approach applies also to a large set of different Petri net classes.

\section{BIGCARL ALGORITHM}\label{sec:algorithm}

The BigCarl algorithm mines frequent subnets from a single large Petri net. A Petri net will be first transformed into an e-type net graph introduced above, on which the subnet mining is done. In this section we only discuss e-type subnet mining. The technique of C-type subnet mining is similar.

\begin{algorithm}

\caption{Mining frequent sub-net graphs ($MinSup$)}
\raggedright

\KwOut{The set of frequent net graphs and their embeddings}
	\LinesNumbered

	\setcounter{AlgoLine}{0}
	Sort the code units in both $MinE[0]$ and $MinE[1]$ according to their alphabetical order;\\
	Remove infrequent nodes from $MinE[0]$ and put the patterns of the remaining nodes (frequent level-0 net graphs) in $MinF[0]$, such that the nodes in $MinE[0,j]$ are the embedding of the $j$-th node pattern in $MinF[0]$;\\
	Remove infrequent edges from $MinE[1]$ and put the patterns of the remaining edges (frequent level-1 net graphs) in $MinF[1]$, such that the edges in $MinE[1,j]$ are the embedding of the $j$-th edge pattern in $MinF[1]$;\\
	Call Algorithm Finding all frequent sub-net graphs (1, $MinF[1],MinE[1],\emptyset,MinF[1],MinE[1],MinSup$);
\end{algorithm}
In the following algorithm, we use a queue $MinFC[k]$ to save the $k$ candidate patterns in a net graph (Line 1). All the embedding copies of these candidate patterns are saved in a hash-map $MinFEC[k,j]$ (Line 2). If two embedding occur in the same $S[h]$, which means they have the condition overlap phenomenon, then an overlap graph can be constructed. Each embedding in the hash-map $MinFEC[k,j]$ is the node, each condition overlap relationship is the edge (Line 3-4). If the maximum independent set support is no less than the frequency threshold $MinSup$, use a queue $MinF[k]$ to save the frequent sub-net graphs of the level-$k$ (Line 5-10). Line 11-21 show the generating procedure of level-$k$ frequent net graph in a level-wise manner. In line 14, for each frequent 1-edge in net graph, add it into the level-$k$ frequent sub-net graph to form the level-(k+1)-candidate sub-net graph. In line 20, an early stop filtering rule is said that if there exist level-$(k+1)$ frequent sub-net graphs, the edges of embedding should no less than $(k+1) \times MinSup$. 

\begin{algorithm}

\caption{Finding all frequent sub-net graphs ($k,MinFC[k],MinFEC[k],S,MinF[1],MinFE[1],MinSup$)}
\raggedright

	\setcounter{AlgoLine}{0}

		\LinesNumbered
\For{$j=1$ to $|MinFC[k]|$}
 {
 \tcp{$|MinFC[k]|$ is the number of candidate level-k patterns}
\For{$m=1$ to $|MinFEC[k,j]|$}{
\tcp{$|MinFEC[k,j]|$ is the number of candidate embedding copies of the j-th level-k pattern}
put the $m$-th embedding $e(m)$ in all $S[h]$ whenever $e(m)$ contains the $h$-th node of the net graph;\\
construct an overlap graph with all embedding in $MinFEC[k,j]$ as nodes and add an edge between ant two nodes if they appear in the same $S[h]$;\\
Call the find largest independent set algorithm to get an independent node set;\\
\If{its size is no less than $MinSup$}{
put the $j$-th pattern of $MinFC[k]$ in $MinF[k]$;\\
put the independent node set in $MinFE[k,j]$;
}
}}
\nl $i=0$;\\
\nl \If{$MinF[k] \not = \emptyset$}{
\nl \For{$j=1$ to $|MinF[k]|$}{\nl \ForEach{pattern $p(1)$ for $MinF[1]$ which can be connected to the $j$-th level-$k$ pattern to form a candidate level-$(k+1)$ pattern $p(k+1)$}{

\nl $i=i+1$;\\
\nl \If{no less than $MinSup$ level-$k$ embedding copies in $MinFE[k,j]$ can be extended with an embedding of $p$ from $MinFE[1]$}{\nl put these embedding copies in $MinFEC[k+1,i]$;\\
\nl remove all other embedding copies from $S$;\\

\nl put $p(k+1)$ in $MinFC[k+1]$;
}}}
}
\nl \If{$(k+1) \times MinSup \le NoE$}{\nl	Call Algorithm Finding all frequent sub-net graphs ($k+1$, $MinFC[k+1],MinFEC[k+1],S,MinF[1],MinFE[1],MinSup$);}	

\end{algorithm}

\section{EXPERIMENTAL EVALUATION}\label{sec:evaluation}
\subsection{A methdology for generating big C/E net}

Since large-scale C/E net resources are not available, the test nets ware generated in a random way. To verify the correctness and efficiency of BigCarl, we have designed a naïve algorithm, called DigCarl, which runs directly on the pure C/E net itself. The results show that the frequent sub-net graphs and subnets obtained by two algorithms are consistent. Our BigCarl algorithm outperforms the DigCarl approach largely. We implemented these two algorithms in C++. All the experiments are conducted on a PC with Intel (R) Core (TM) i7-6700 CPU@3.40HZ and 32G RAM, and the operating system is Windows 10.

\begin{definition}[BIG C/E net]
A basic C/E net consists of an event node with some input/output condition nodes connected to it.\\
\end{definition}
$\hfill \square$

The following algorithms extend the Petri net generation functions of \cite{lu_pspan_nodate}.
\begin{algorithm}
\caption{Generate an experimental C/E net ($x, U, H$)}
\raggedright
	\setcounter{AlgoLine}{0}

		\LinesNumbered
		Input a large (small) positive integer $U$ as the maximum of event nodes if a large (small) net has to be generated;\\
 Input a large (small) positive integer $H$ as the maximum of arcs connecting events with conditions if a dense (sparse) net has to be generated;\\
  Input $x = 0 (x = 1)$ if a c-type (e-type) overlap is to be used when connecting two C/E nets; \\
 Create a first basic C/E net; \\
 Call algorithm 5 $U$ times, each time generate a basic C/E net connecting the existing net with at most $H$ arcs by using the x-type overlap approach. 
\end{algorithm}

\begin{algorithm}
\caption{connect$(x,H,N_1,N_2)$}
\raggedright
	\setcounter{AlgoLine}{0}

		\LinesNumbered
	\If{$x=0 (1)$}{
 connect net $N_1$ with $N_2$  with at most $H$ random c-type (e-type) node overlaps;}
\end{algorithm}

In the above algorithms, a large net is generated as target of frequent subnet mining. The generated small nets will be inserted in the large net to test the BigCarl algorithm. Only c-type overlap is used for net connection, which plays also a role in frequent net graphs embedding processing of algorithm 2-3. However, when small nets are inserted (embedded) in a large net, both types of overlap are used. Since algorithms 2-3 only care e-overlap of embedding, we have included e-type overlap in algorithm 4-5 just for showing that frequent subnet mining with e-type overlap is also possible.

\begin{example}
Reconsider Figure ~\ref{1} in a reverse way as a net generation process. Then \ref{b} (\ref{c}) is an e-type (c-type) basic net. \ref{d}(\ref{e})is a net connection with c-type(e-type)overlap. 
\end{example}

\subsection{Testing with planted nets}

Inspired by the planted motif approaches in frequent subgraph mining \cite{yang_combining_2009}, we design a planting strategy to test the correctness of BigCarl. First, we use the following terminology: (1) Planting net. The c-connect (e-connect) copies are generated based on these nets. (2) Big Planted net. The net where the copies are already planted in. (3) Big\_net. The incremental net that is waiting for the connection of big planted nets. The details are depicted as follows.

\begin{algorithm}[H]\label{plant}
\caption[position=top]{Planting-nets($NS,g,H,MinSup$)}
\raggedright
	\setcounter{AlgoLine}{0}

		\LinesNumbered
	\KwIn{A set of planting nets $NS$, two parameters $g,H$, which denote the maximum of events in each planted net and the maximum of arcs connected with each event, the frequency threshold $0<Minsup<N$;}
	Call algorithm 4 to generate a test net under $g,H$ constraints;\\
	\ForEach{planting net $s \in NS$}{Generate two random integers $c(s),e(s)$ such that $Minsup<c(s)\le N,Minsup<e(s) \le N$; \\
	 Generate $c(s)$ c-connect copies of $s$ randomly, and put them into $S$;\\
	 Generate $e(s)$ e-connect copies of $s$ randomly, and put them into $S$; }
 Shuffle $S$;\\ \tcp{disorder all the copies randomly}
 \ForEach{selected disordered copy $y \in S$}{Call connect $(0,x, y)$ or connect $(1,x, y)$ to produce a big test net $x [+] y$;}
\end{algorithm}
The validation rule is that for each big planted big nets in $NS$, perform BigCarl algorithm and check whether the mined results contain the planting net $s$ and their copies or not. The detail can be illustrated as algorithm \ref{valid}. The testing result be can seen in table \ref{MIS}.

\begin{algorithm} [H]
\caption{Generate frequent subnet in a single large net with predefined overlap graph}

\raggedright
	\setcounter{AlgoLine}{0}
\label{valid}
		\LinesNumbered
	Given $m$ copies of a small net $N'$ where the size of each small net is no less than $m$;\\ \tcp{make sure the enough overlap}
 Randomly generate $n$ pairs of natural numbers (denoting pairs of small net copies) where each natural number is between 1 and $m(m-1)$; \\
 Insert the $m$ copies of $N'$ in $N$ one by one;\\
 When the $j$-th copy is inserted, it must have a c-overlap with each $h$-th copy for $h<j$;\\\tcp{The $h$-th copy is inserted in $N$ earlier than $j$-th copy does}
 Check whether $(h,j)$ is a pair generated above;\\ \tcp{This means they should have a c-overlap}
\end{algorithm}
\begin{example}
Suppose the initial test net has 3K nodes, 6K edges. We use the small net which has 100 nodes, 200 edges to generate the copies. Let $MinSup=50$, suppose each small net has $m=10 (m< MinSup)$ copies.  
\end{example}
\begin{table}
  \caption{The procedure of building overlap graph and calculate the maximum independent set}
\label{MIS}

  \begin{tabular}{cclcc}
    \toprule
   \tabincell{c}{level-k}&\tabincell{c}{\#Embedding} &\tabincell{c}{Edges in\\overlap\\graph}&\tabincell{c}{Maximal\\independent\\set}&Runtime (s)\\
    \midrule
    1 & 780&  187&410&1820\\
    2 & 408&  100&280&890\\
    3 & 276 & 56&172&265\\
    4 &115& 27&67&87\\
    5 & 41&10&23&38\\
  \bottomrule
\end{tabular}
\end{table}

\subsection{Performance of BigCarl}

In this section, we will verify the results with two experiments. (1) The reduction power of net graph structures. When the number of arcs increasing, the ratio of net graph edges’ number/ big C/E net arcs’ number reduces exponentially. (2) The scalability of BigCarl vs. DigCarl.

\subsubsection{The reduction power of net graph structures}

Under the same parameters of net generating, we generate 5 big C/E nets with a different scale. We investigate the change of the ratio (e-type net graph edges’ number/big C/E net arcs’ number) by increasing the number of nodes in big C/E nets (10K,28K,58K,162K,490K, including events, conditions) and increasing the number of arcs as shown in Table 2. The reduction ratio can be seen in table \ref{t} and figure \ref{compress}.

\begin{table}
  \caption{The edges number of e-type net graph vs. the arcs number of big C/E nets}
  \label{t}
  \scriptsize
  \begin{tabular}{cclccc}
    \toprule
   \tabincell{c}{ \#arcs in big\\  C/E nets (AB)}&\tabincell{c}{5400}
    &11,518&23,416&64,811&281,575\\
    \midrule
    \tabincell{c}{\#edges in e-type\\ net graph (AE)} & 1522& 3196&4699&6971&25420\\
     \midrule
    AE/AB & 28\%&27.7\%&20\%&10.7\%&9\%\\
 
  \bottomrule
\end{tabular}
\end{table}
\begin{figure}
\includegraphics[width=.6\textwidth]{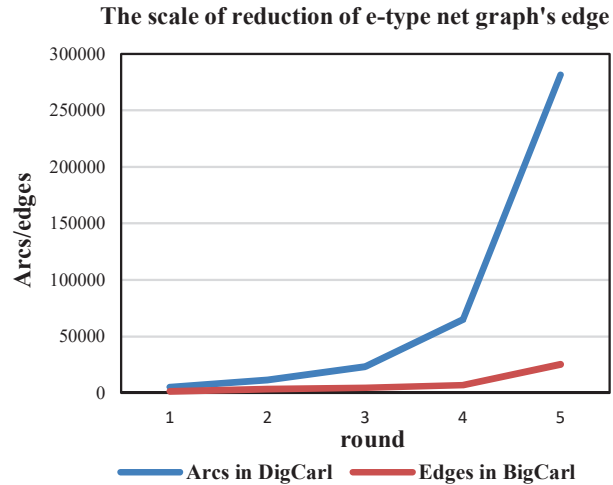}
\caption{The reduction of e-type net graph’s edge}
\label{compress}
\end{figure}

\subsubsection{The scalability of BigCarl vs. DigCarl}
The following experiments compare the salability of BigCarl with DigCarl by evaluating (1) the runtime (in seconds), (2) the memory overheads (in MB) depending on the growing number of big C/E net arcs. Note that due to the hardware limitations, the big C/E nets (with almost 10K nodes, 20K edges) cannot be mined in a tolerable time. Therefore, we mainly focus on the big C/E nets with maximum 5K nodes and 10K edges to compare the scalability. Suppose the frequency threshold is 10, varying the number of nodes (2835, 4163, 4927,5561,6019) and the number of arcs (6K,7K,8K,9K,10K). The trend of runtime and memory usage can be seen in figure \ref{ra} and \ref{rb}. It shows that the overheads of BigCarl are significantly reduced compared with those of DigCarl. Figure \ref{ra} shows that the growth trend of both BigCarl and DigCarl is exponential, whereas DigCarl has been running for more than 1 day (it cannot be shown in the figure). 

\begin{figure}
%

	\begin{minipage}[c]{\linewidth}
		\centering
		\subfloat[Runtime ]{\includegraphics[width=0.4\textwidth]{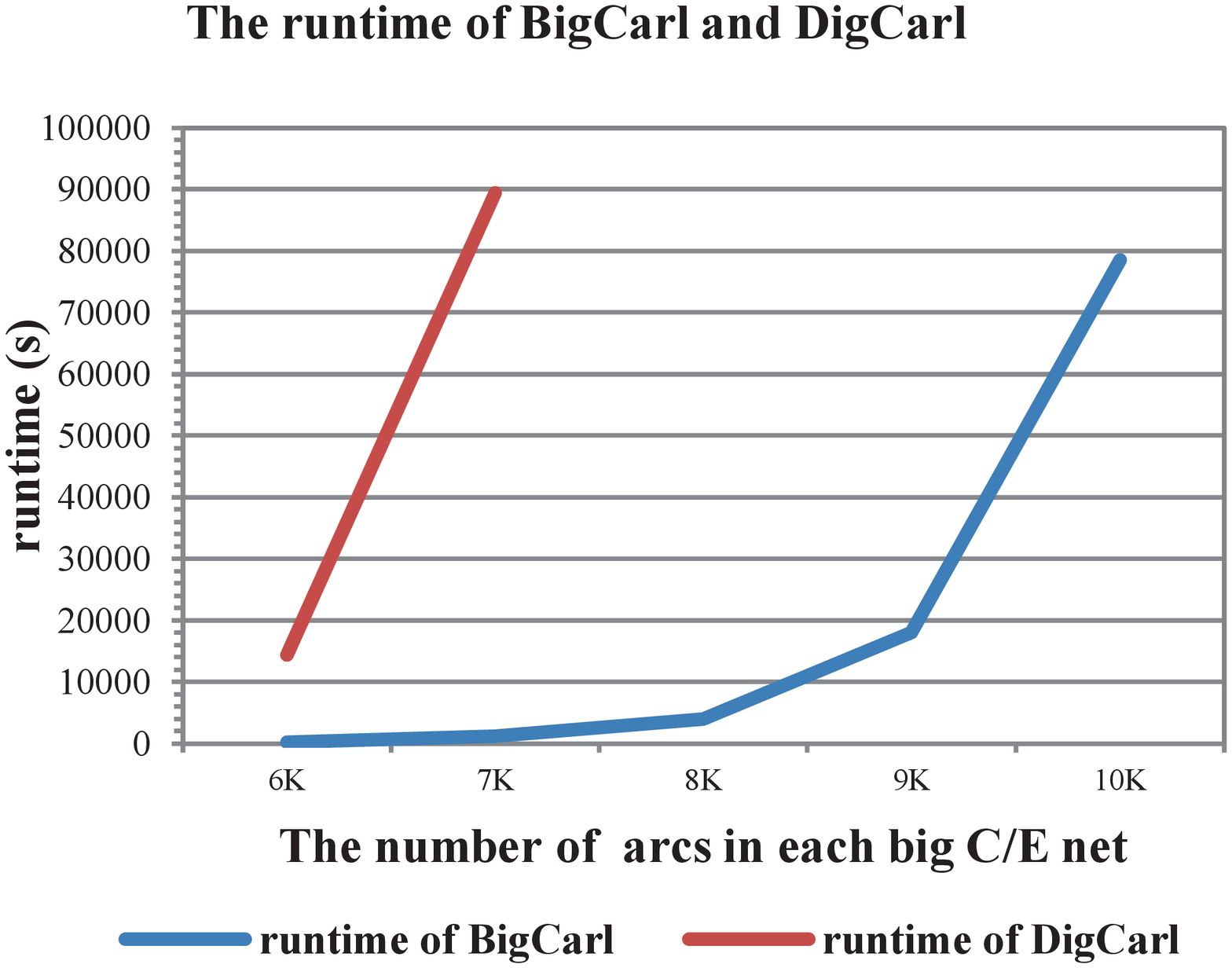}\label{ra}}

	\end{minipage}
		\begin{minipage}[c]{\linewidth}
			\centering
			\subfloat[Memory ]{\includegraphics[width=0.4\textwidth]{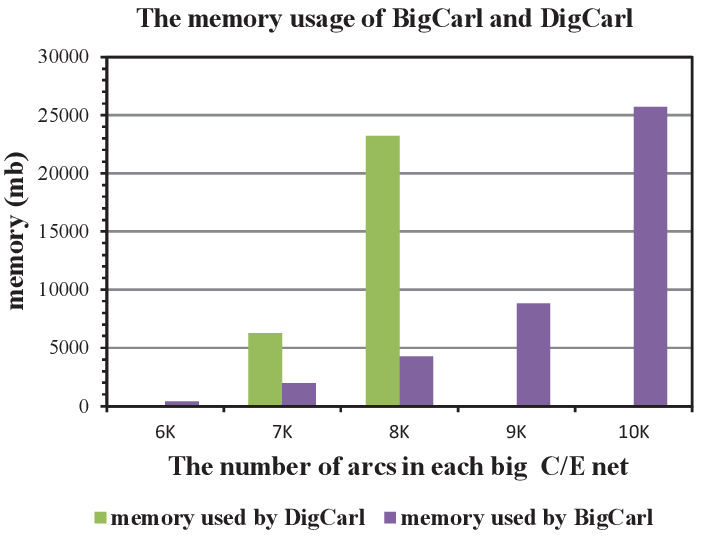}\label{rb}}

		\end{minipage}
			

\caption{The comparison of BigCarl and DigCarl}
\label{3}
\end{figure}
 
\section{CONCLUSION}\label{sec:conclude}
The main contributions of this paper can be summarized as follows: This paper is the first one addressing the problem of finding frequent subnets from a single large Petri net. We not only make use of Petri net’s basic properties in algorithm design but have extended them with a duality principle about Petri nets’ event and condition nodes. With this principle we differentiate between c-type and e-type (net graph, overlap, connection, basic net, complete net, frequent subnet mining) and thus expanded the application possibility of our research results largely. We established a detailed framework for implementing the MIS algorithm in case of Petri nets and their net graph representation. To evaluate the effectiveness of our approach we developed a technique of frequent subnet mining under predefined overlap graph. This technique inserts copies of small Petri nets in a large Petri net with randomly generated overlap schema. This technique provides the possibility of checking the correctness of frequent subnet mining. Our BigCarl algorithm can mine 10K nodes scale Petri nets transform 100K nodes scale Petri nets in net graphs.
  \bibliographystyle{ACM-Reference-Format}
  \bibliography{review}


\end{document}